\documentclass[11pt]{article}

\usepackage[T1]{fontenc}
\usepackage[utf8]{inputenc}
\usepackage[margin=1in]{geometry}
\usepackage{graphicx}
\usepackage{amsmath,amssymb,amsfonts}
\usepackage{amsthm}
\usepackage{booktabs,tabularx}
\usepackage{microtype}
\usepackage{hyperref}
\hypersetup{colorlinks=true,linkcolor=blue,citecolor=blue,urlcolor=blue}
\theoremstyle{plain}
\newtheorem{proposition}{Proposition}[section]
\DeclareMathOperator{\col}{col}
\newcolumntype{Y}{>{\raggedright\arraybackslash}X}

\title{Theoretical--operational modelling of complex experiments:\\
parameter robustness and degeneracy in muon--electron conversion}
\author{Vitaly Pronskikh\\
\small Fermi National Accelerator Laboratory, Batavia, Illinois, USA\\
\small \texttt{vpronskikh@gmail.com}}
\date{}

\begin{document}
\maketitle

\begin{abstract}
Complex experiments infer theory parameters through a coupled chain of physical
models and data reduction.  We formulate the theoretical--operational model
(TOM) as a typed factorization of this forward prediction and study how
changes of preparation, phenomenon modelling, readout, backgrounds, and
analysis project onto the local manifold generated by the physics parameters.
For a smooth prediction and a locally identifiable weighted least-squares
estimate, the resulting response map separates each model deformation into a
parameter-equivalent component and a residual component that cannot be
absorbed by a change of the fitted physics parameters.  This gives local
criteria for robustness, exact degeneracy, and partial degeneracy of parameter
inference.  The construction is applied to charged-lepton-flavour-violating
muon--electron conversion in aluminium.  A one-bin conversion-rate model
exhibits an exact normalization degeneracy.  In a two-template model of the
elastic--inelastic spectrum, a common signal normalization is absorbed by the
fitted conversion rate without changing an operator-sensitive nuclear-response
ratio, whereas a relative elastic--inelastic efficiency change is exactly
parameter-equivalent to a change of that ratio at first order.  A numerical
Run-I example based on published Mu2e spectra shows that a
\(100\,\mathrm{keV}/c\) momentum-scale mismatch is only partially
parameter-equivalent: its projection ratio onto the local tangent space
generated by the elastic normalization \(R_0\) and the logarithmic
inelastic-to-elastic response ratio \(\rho\) is \(\eta=0.331\), while most of the
weighted spectral deformation remains as a residual shape.  TOM thereby
provides a local theoretical description of parameter robustness and
degeneracy at the interface of particle/nuclear phenomenology and
experimental realization.
\end{abstract}


\section{Introduction}
\label{sec:intro}

A precision experiment confronts a physical theory with data only through a
structured forward model.  A source or initial state is prepared, the
phenomenon of interest acts on that input, a detector converts the result into
recorded data, and an analysis procedure constructs the quantities used for
inference.  The resulting prediction depends simultaneously on physics
parameters and on auxiliary choices describing the realization of the
experiment.  A theoretical question therefore arises before any particular
statistical treatment is chosen: which changes of that realization are locally
equivalent, in observable space, to changes of the physics parameters, and
which changes remain distinguishable?

The theoretical--operational model (TOM) develops the operational viewpoint
discussed previously in Ref.~\cite{Pronskikh2020} into a local theory of
parameter response.  Its starting point is a typed factorization of the physical forward model.
Schematically,
\begin{equation}
 \mathcal P\longrightarrow\mathcal X_\theta\longrightarrow\mathcal M
 \longrightarrow\lambda_{\rm sig},\qquad
 \lambda_{\rm tot}=\lambda_{\rm sig}+\sum_b\lambda_b
 \xrightarrow{\ A\ }\boldsymbol\mu\longrightarrow\widehat\theta.
 \label{eq:tom-overview}
\end{equation}
Here \(\mathcal P\) is preparation, \(\mathcal X_\theta\) is the model of the
specified phenomenon with physics parameters \(\theta\), \(\mathcal M\) is
detector readout, and \(A\) is classical analysis.  The finite measure
\(\lambda_{\rm sig}\) is the expected record generated by the phenomenon branch,
while separate branches \(\lambda_b\) describe physical backgrounds that can
mimic it.  The complete expected record \(\lambda_{\rm tot}\) is mapped by the
analysis to the scalar or vector prediction \(\boldsymbol\mu\) used to obtain
\(\widehat\theta\).

Preparation, phenomenon, readout, analysis, and estimation have different
domains and codomains; TOM unifies them through the typed composition in
Eq.~\eqref{eq:tom-overview}.  This functional separation has antecedents in
operational formulations of quantum mechanics \cite{Fock1957,Busch1995}.  Here
it is used to study the local relation between deformations of the full forward
model and displacements of the inferred theory parameters.

Source, detector, calibration, background, and analysis uncertainties are
represented by nuisance parameters in standard experimental statistics.  TOM
focuses on the geometry of the corresponding local response: the columns of
the physics Jacobian span the observable-space directions generated by changes
of the theory parameters, while deformations of the remaining model generate
additional directions that may lie inside, outside, or partly inside that
span.  This distinguishes changes that are parameter-equivalent from changes
that leave a residual not reproducible by varying the physics parameters.  The
construction is closely related to local-influence and sensitivity calculations
\cite{Cook1986}; the TOM factorization retains the physical origin of each
deformation while making its relation to the theory-parameter manifold
explicit.

The phenomenon theory also constrains the admissible realization of the
experiment: it determines which parameters are of interest, which preparations
are relevant, and which features of the final state must be distinguished.
The arrows in Eq.~\eqref{eq:tom-overview} describe the forward flow of states
and records.  Preparation and measurement may still depend on the theory being
tested.

Section~\ref{sec:formalism} defines the typed forward model.
Section~\ref{sec:sensitivity} derives the local parameter-response map and its
degeneracy structure.  Section~\ref{sec:mu2e} applies the construction to
elastic and inelastic \(\mu^-\!\to e^-\) conversion in \(^{27}\mathrm{Al}\),
where the inferred quantities are tied to CLFV and nuclear-response
phenomenology.

\section{Theoretical--operational forward model}
\label{sec:formalism}

\subsection{Typed composition}
\label{subsec:stages}

Let \(\Theta\subset\mathbb R^d\) be the parameter space of the phenomenon
model, let \(\mathcal S_{\rm in}\) and \(\mathcal S_{\rm out}\) be its input
and output description spaces, and let \(\Omega\) be the space of recorded
outcomes.  Denote by \(\mathfrak M_+(\Omega)\) the finite nonnegative measures
on \(\Omega\).  The preparation specifies an exposure \(N\) and an input state
or ensemble \(s_P\).  The four stages have the types
\begin{equation}
 \begin{aligned}
 \mathcal P &:U_P\longrightarrow\mathbb R_+\!\times\mathcal S_{\rm in},
 &\mathcal P(u_P)&=\bigl(N(u_P),s_P(u_P)\bigr),\\
 \mathcal X_{\theta,u_X}&:\mathcal S_{\rm in}\longrightarrow\mathcal S_{\rm out},
 &\mathcal M_{u_M}&:\mathcal S_{\rm out}\longrightarrow\mathfrak M_+(\Omega),\\
 A_{u_A}&:\mathfrak M_+(\Omega)\longrightarrow\mathbb R^k.&&
 \end{aligned}
 \label{eq:tom-types}
\end{equation}
When \(\mathcal S_{\rm in}\) and \(\mathcal S_{\rm out}\) are quantum state
spaces, \(\mathcal X_{\theta,u_X}\) may be realized by a quantum operation,
while \(\mathcal M_{u_M}\) is the outcome-law map induced by a
positive-operator-valued measure (POVM).  If post-measurement states are part
of the description, a quantum instrument supplies the corresponding
refinement \cite{Busch1995,Watrous2018}.  The developments below use only the
resulting finite measure on \(\Omega\), so they do not depend on this quantum
specialization.

The collective model coordinates are
\[
 u=(u_P,u_X,u_M,u_A,u_{\rm bg})\in U,
\]
where \(U\) is the local coordinate domain for the auxiliary stage variables.
Here \(u_P,u_M,u_A\) describe the experimental realization and analysis,
whereas \(u_X\) denotes phenomenon-model inputs not selected as parameters of
interest.  The latter may be fixed, constrained, profiled, or marginalized.
The background coordinate block \(u_{\rm bg}\) collects the individual \(u_b\) and is absent when no explicit
background branches are required.

The expected signal record is
\begin{equation}
 \lambda_{\rm sig}(\cdot\mid\theta,u_P,u_X,u_M)
 =N(u_P)\,
 \mathcal M_{u_M}\!\left[
   \mathcal X_{\theta,u_X}\!\left(s_P(u_P)\right)
 \right].
 \label{eq:signal-record}
\end{equation}
Each mimicking physical process is represented by a separate expected record
measure \(\lambda_b(\cdot\mid u_P,u_M,u_b)\).  The complete record and analysed
prediction are
\begin{equation}
 \lambda_{\rm tot}
 =\lambda_{\rm sig}+\sum_b\lambda_b,
 \qquad
 \boldsymbol\mu(\theta,u)=A_{u_A}[\lambda_{\rm tot}].
 \label{eq:tom-forward}
\end{equation}
For a binned count analysis, \(A\) simply integrates the measure over the
selected bins, so \(\mu_i=\lambda_{\rm tot}(\Delta_i)\) and signal and background
counts add linearly.  More elaborate classical summaries can be represented by
other choices of \(A\).

Stage assignment is determined by physical role, not merely by temporal order.
All source, transport, interaction, and stopping processes needed to create
\(s_P\) belong to preparation.  The map \(\mathcal X\) is reserved for the
phenomenon whose physics parameters are being inferred.  Interactions that turn
its output into a detector record belong to readout, whereas distinct physical
processes capable of producing a mimicking record are backgrounds.  This separation is functional.

\subsection{Model coordinates and admissible deformations}
\label{subsec:stage-changes}

A reference realization is specified by \(u_0\).  A small continuous change is
written
\[
 \delta u=(\delta u_P,\delta u_X,\delta u_M,\delta u_A,
           \delta u_{\rm bg}).
\]
It may represent a change in source conditions, phenomenon-model inputs,
detector calibration or response, analysis choices, or a background model.
When one physical quantity affects several stages, its total effect is obtained
by differentiating through every affected stage; it is not duplicated as
independent nuisance parameters.  Discrete changes, such as replacing a
reconstruction algorithm, are compared by finite differences rather than by an
infinitesimal derivative.

\section{Parameter response, robustness, and degeneracy}
\label{sec:sensitivity}

\subsection{Local weighted least-squares estimator}
\label{subsec:local-estimator}

Let
\begin{equation}
 \mu:\Theta\times U\longrightarrow\mathbb R^k
 \label{eq:prediction-map}
\end{equation}
be the prediction defined by the complete TOM chain.  At a nominal point
\((\theta_0,u_0)\), assume that (i) the chosen coordinate domains are open in a
neighbourhood of \((\theta_0,u_0)\), (ii) \(\mu\) is twice continuously differentiable, (iii)
\(y_0=\mu(\theta_0,u_0)\), (iv) \(W\) is fixed, symmetric, and positive
definite, and (v)
\begin{equation}
 J_\theta
 =\left.\frac{\partial\mu}{\partial\theta}\right|_{(\theta_0,u_0)}
 \label{eq:physics-jacobian}
\end{equation}
has full column rank \(d\).  The last condition is local identifiability with
the model coordinates fixed.  If the original parameterization contains
unidentifiable directions, it must first be replaced locally by identifiable
coordinates.

For \(y\) and \(u\) close to their nominal values define
\begin{equation}
 Q(\vartheta;y,u)
 =\frac12[y-\mu(\vartheta,u)]^{\top} W[y-\mu(\vartheta,u)].
 \label{eq:wls-objective}
\end{equation}
The estimator \(\widehat\theta(y,u)\) is the strict local minimizer on the branch
through \(\widehat\theta(y_0,u_0)=\theta_0\).  The construction is local to this
branch and makes no global existence or uniqueness claim.

\begin{proposition}[Local estimator response]
\label{prop:local-estimator-response}
Under the assumptions above, there are neighbourhoods of \((y_0,u_0)\) and
\(\theta_0\) in which the stationary branch \(\widehat\theta(y,u)\) is unique,
is a strict local minimizer, and is continuously differentiable.  Its
first-order response is
\begin{equation}
 \widehat\theta(y_0+\delta y,u_0+\delta u_{\rm fit})-\theta_0
 =J_{\theta,W}^{+}\left(\delta y-J_u\delta u_{\rm fit}\right)
 +o\!\left(\|\delta y\|+\|\delta u_{\rm fit}\|\right),
 \label{eq:master-estimator-response}
\end{equation}
where
\begin{equation}
 J_u=\left.\frac{\partial\mu}{\partial u}\right|_{(\theta_0,u_0)},
 \qquad
 J_{\theta,W}^{+}
 =\left(J_\theta^{\top}WJ_\theta\right)^{-1}J_\theta^{\top}W.
 \label{eq:weighted-left-inverse}
\end{equation}
\end{proposition}

\begin{proof}
The stationarity equation for Eq.~\eqref{eq:wls-objective} is
\begin{equation}
 g(\vartheta;y,u)
 =D_\theta\mu(\vartheta,u)^{\top}W[\mu(\vartheta,u)-y]=0.
 \label{eq:wls-stationarity}
\end{equation}
At the nominal exact fit, \(g(\theta_0;y_0,u_0)=0\), and
\[
 \left.\frac{\partial g}{\partial\vartheta}\right|_{(\theta_0;y_0,u_0)}
 =J_\theta^{\top}WJ_\theta,
\]
which is positive definite by the rank assumption.  The implicit-function
theorem gives a unique differentiable stationary branch, and continuity of the
Hessian makes it a strict local minimum in a sufficiently small neighbourhood.
Differentiating Eq.~\eqref{eq:wls-stationarity} at the nominal point gives
\[
 J_\theta^{\top}W\left(J_\theta\,\delta\widehat\theta
      +J_u\delta u_{\rm fit}-\delta y\right)=0,
\]
which yields Eq.~\eqref{eq:master-estimator-response}.
\end{proof}

For a fixed observed summary, reanalysis with
\(u_{\rm fit}=u_0+\delta u_{\rm fit}\) therefore gives
\begin{equation}
 \delta\widehat\theta_{\rm fixed\ summary}
 =-J_{\theta,W}^{+}J_u\delta u_{\rm fit}
 +o(\|\delta u_{\rm fit}\|).
 \label{eq:fixed-data-shift}
\end{equation}
If changing \(u_A\) also changes the summary computed from the raw data, that
effect belongs in \(\delta y\); Eq.~\eqref{eq:master-estimator-response} keeps
the two contributions separate.

\subsection{Mismatch between realized and fitted experiments}
\label{subsec:true-fit-mismatch}

Suppose that the expected summary is generated at
\(u_{\rm true}=u_0+\delta u_{\rm true}\), while the prediction used in the fit
assumes \(u_{\rm fit}=u_0+\delta u_{\rm fit}\).  For an Asimov data set, or for
the systematic part of the expected summary,
\[
 \delta y=J_u\delta u_{\rm true}+o(\|\delta u_{\rm true}\|).
\]
Substitution into Eq.~\eqref{eq:master-estimator-response} gives
\begin{equation}
 \Delta\theta_{\rm sys}
 =J_{\theta,W}^{+}J_u
   (\delta u_{\rm true}-\delta u_{\rm fit})
 +o\!\left(\|\delta u_{\rm true}\|+\|\delta u_{\rm fit}\|\right).
 \label{eq:true-fit-parameter-shift}
\end{equation}
Thus the sign is determined by the mismatch between the realized and fitted
stage descriptions.  If both change consistently, there is no first-order
nominal bias, although the statistical precision may still change.  A
statistical fluctuation \(\epsilon=y-\mathbb E[y]\) adds the local term
\(J_{\theta,W}^{+}\epsilon\).  Equation~\eqref{eq:true-fit-parameter-shift} is
not a claim about exact finite-sample bias.

\subsection{Theoretical--operational response map}
\label{subsec:stage-resolved}

The stage Jacobian inherits the typed decomposition,
\begin{equation}
 J_u=\begin{bmatrix}J_P&J_X&J_M&J_A&J_{\rm bg}\end{bmatrix},
 \qquad
 J_s=\left.\frac{\partial\mu}{\partial u_s}\right|_{(\theta_0,u_0)}.
 \label{eq:stage-jacobians}
\end{equation}
Define the component-resolved parameter-response matrices
\begin{equation}
 B_s=J_{\theta,W}^{+}J_s,
 \qquad
 B=\begin{bmatrix}B_P&B_X&B_M&B_A&B_{\rm bg}\end{bmatrix}
 =J_{\theta,W}^{+}J_u.
 \label{eq:stage-parameter-response}
\end{equation}
With \(v_s=\delta u_{s,{\rm true}}-\delta u_{s,{\rm fit}}\),
Eq.~\eqref{eq:true-fit-parameter-shift} becomes
\begin{equation}
 \Delta\theta_{\rm sys}
 =B_Pv_P+B_Xv_X+B_Mv_M+B_Av_A+B_{\rm bg}v_{\rm bg}
 +o(\|v\|).
 \label{eq:stagewise-bias}
\end{equation}
The matrix \(B=J_{\theta,W}^{+}J_u\) is the principal local response map of
TOM: it maps deformations of the theoretical--operational model to
displacements in the inferred physics-parameter space.  Its blocks \(B_s\)
resolve that map according to the physical origin of the deformation---
preparation, phenomenon modelling, readout, analysis, or background.  A shared
coordinate is treated with the total derivative through every component it
affects, not as independent copies in several blocks.

\subsection{Robustness of the forward and inference maps}
\label{subsec:robustness-levels}

For a mismatch direction
\(v=\delta u_{\rm true}-\delta u_{\rm fit}\), three local notions of robustness
should be distinguished.  \textbf{Record-level robustness} means that, for a
direction that changes the physical or background stages but not the analysis
map,
\begin{equation}
 D_{(u_P,u_X,u_M,u_{\rm bg})}\lambda_{\rm tot}[v_P,v_X,v_M,v_{\rm bg}]=0.
 \label{eq:record-robustness}
\end{equation}
The complete expected record, including exposure and backgrounds, is then
unchanged to first order.  \textbf{Summary-level robustness} requires
\begin{equation}
 J_uv=0,
 \label{eq:summary-robustness}
\end{equation}
so that the analysed prediction is unchanged.  \textbf{Estimator-level
robustness} requires only
\begin{equation}
 Bv=J_{\theta,W}^{+}J_uv=0.
 \label{eq:estimator-robustness}
\end{equation}
For a parameter combination \(L\theta\), the corresponding condition is
\(LBv=0\).

For a fixed analysis map, record-level robustness implies summary-level
robustness, which in turn implies estimator-level robustness.  The converses do
not generally hold.  If \(u_A\) changes, the expected record may be unchanged
while the summary changes.  Estimator-level robustness is therefore relative to
the selected summary, estimator, weight matrix, and nominal point.  It is a
local sensitivity property of the chosen inference setup.

\subsection{Exact and partial parameter-equivalent deformations}
\label{subsec:operational-degeneracy}

A parameter response can occur without exact degeneracy.  Let
\begin{equation}
 q=J_uv
 \label{eq:stage-summary-direction}
\end{equation}
be the summary-space direction produced by a model mismatch.  Define the
\(W\)-orthogonal projector onto the parameter-sensitive column space by
\begin{equation}
 \Pi_\theta
 =J_\theta(J_\theta^{\top}WJ_\theta)^{-1}J_\theta^{\top}W
 =J_\theta J_{\theta,W}^{+}.
 \label{eq:parameter-projector}
\end{equation}
Then
\begin{equation}
 q=q_\parallel+q_\perp,
 \qquad
 q_\parallel=\Pi_\theta q,
 \qquad
 q_\perp=(I-\Pi_\theta)q.
 \label{eq:stage-decomposition}
\end{equation}
A nonzero deformation is locally parameter-equivalent at first order when its
summary-space direction satisfies
\begin{equation}
 q_\perp=0,
 \qquad\text{equivalently}\qquad
 q\in\col(J_\theta)\setminus\{0\}.
 \label{eq:exact-degeneracy}
\end{equation}
If both components are nonzero, the parameter equivalence is only partial: the
parallel component is reproduced by a displacement of the physics parameters,
while the perpendicular component remains as a residual outside the local
theory-parameter tangent space.  If
\(q_\parallel=0\) but \(q\neq0\), the summary changes while the local estimate is
robust.

For \(q\neq0\), define the projection ratio
\begin{equation}
 \eta(v)=\frac{\|\Pi_\theta q\|_W}{\|q\|_W},
 \qquad
 \|z\|_W=(z^{\top}Wz)^{1/2},
 \qquad 0\leq\eta\leq1.
 \label{eq:degeneracy-ratio}
\end{equation}
Thus \(\eta=1\) denotes exact local parameter equivalence, \(0<\eta<1\)
partial equivalence, and \(\eta=0\) a deformation orthogonal to the local
theory-parameter tangent space.  The magnitude of the parameter shift also depends on the conditioning of
\(J_\theta\); \(\eta\) therefore quantifies alignment, not the size of the bias.

This construction is directly related to the information blocks used in
nuisance-parameter analysis.  For fixed positive-definite \(W\),
\begin{equation}
 G_{\theta\theta}=J_\theta^{\top}WJ_\theta,
 \qquad
 G_{\theta u}=J_\theta^{\top}WJ_u,
 \qquad
 B=G_{\theta\theta}^{-1}G_{\theta u}.
 \label{eq:fisher-block-relation}
\end{equation}
For a local Gaussian model with \(W=\Sigma^{-1}\), these are the usual
information blocks; for independent Poisson bins, the corresponding Asimov
curvature uses \(W=\operatorname{diag}(1/\mu_i)\).  If model coordinates are
fitted jointly with \(\theta\), identifiability must instead be assessed from
the combined Jacobian together with the associated auxiliary constraints
\cite{Kay1993,CowanEtAl2011}.

\section{Application to charged-lepton flavour violation: muon--electron conversion in aluminium}
\label{sec:mu2e}

\subsection{Phenomenological setting and theoretical--operational chain}
\label{subsec:mu2e-question}

Mu2e searches for charged-lepton-flavour-violating (CLFV) conversion of a bound
negative muon into an electron in the field of an aluminium nucleus.  The
primary search is for coherent, elastic conversion,
\begin{equation}
 \mu^-+{}^{27}\mathrm{Al}\longrightarrow e^-+{}^{27}\mathrm{Al},
 \label{eq:elastic-conversion}
\end{equation}
whose electron momentum is close to \(104.97\,\mathrm{MeV}/c\).  The standard
rate is
\begin{equation}
 R_{\mu e}
 =\frac{\Gamma(\mu^-A\to e^-A)}{\Gamma(\mu^-A\to\nu_\mu A')},
 \label{eq:Rmue}
\end{equation}
where the denominator is the ordinary nuclear-capture rate
\cite{KunoOkada2001,Kitano2002,Mu2eTDR2015}.

At nuclear level the conversion rates may be written schematically as
\begin{equation}
 R_f(\boldsymbol c,u_X)=\boldsymbol c^\dagger H_f(u_X)\boldsymbol c,
 \label{eq:nuclear-response}
\end{equation}
where \(\boldsymbol c\) is a vector of low-energy CLFV couplings and \(H_f\)
is constructed from nuclear responses.  A single target-dependent conversion
rate cannot identify a general vector \(\boldsymbol c\)
\cite{Cirigliano2009}.  Additional operator-sensitive information can appear
in conversion to low-lying excited states of \(^{27}\mathrm{Al}\).  The first
three relevant levels have excitation energies \(0.844\), \(1.015\), and
\(2.212\,\mathrm{MeV}\), corresponding to conversion-electron momenta of
approximately \(104.13\), \(103.96\), and \(102.77\,\mathrm{MeV}/c\)
\cite{HaxtonRule2024,HaxtonRule2025}.  Nuclear effective theory predicts
operator-dependent elastic and inelastic response patterns
\cite{RuleEtAl2023,HaxtonEtAl2023,HaxtonRule2024,HaxtonRule2025}.  For
hypotheses with a nonzero elastic rate we therefore use the response ratios
\begin{equation}
 r_f(\boldsymbol c,u_X)=\frac{R_f(\boldsymbol c,u_X)}{R_0(\boldsymbol c,u_X)},
 \qquad \rho_f=\log r_f.
 \label{eq:inelastic-ratio}
\end{equation}
Inelastic-only responses, for which the elastic denominator vanishes by a
selection rule, require separate absolute amplitudes rather than these ratios.
The TOM question is which deformations of the complete forward model are locally
equivalent to changes of an operator-sensitive response ratio, and which leave
a distinguishable spectral residual.

The public Mu2e Run-I projection provides a realistic experimental realization
for this question.  It uses \(6\times10^{16}\) stopped muons, with optimized
selection \(103.60<p_{\rm rec}<104.90\,\mathrm{MeV}/c\) and
\(640<T_0<1650\,\mathrm{ns}\), a single-event sensitivity of
\(2.4\times10^{-16}\), and total expected background
\(0.105\pm0.032\) events \cite{Mu2eRunI2023}.  In the corresponding one-bin
approximation,
\begin{equation}
 \mu_{\rm win}(R_{\mu e})
 \simeq\frac{R_{\mu e}}{2.4\times10^{-16}}+0.105.
 \label{eq:run1-one-bin}
\end{equation}
The numerical example below uses these projected spectra as a benchmark for
local parameter degeneracy; it is not an official Mu2e sensitivity result.

For this problem the signal record is
\begin{equation}
 \lambda_{\rm sig}^{\mu e}
 =N_\mu^{\rm stop}(u_P)\,
 \mathcal M_{u_M}\!\left[
  \mathcal X^{\mu e}_{\theta,u_X}(s_{\mu{\rm Al}}(u_P))
 \right],
 \qquad
 \boldsymbol\mu
 =A_{u_A}\!\left[\lambda_{\rm sig}^{\mu e}+\sum_b\lambda_b\right].
 \label{eq:mu2e-chain}
\end{equation}
The phenomenon map \(\mathcal X^{\mu e}\) represents only the neutrinoless
conversion \(\mu^-+{}^{27}{\rm Al}\to e^-+{}^{27}{\rm Al}^{(*)}\).
Production-target interactions, pion production and decay, muon transport,
stopping, atomic binding, and the cascade into a bound muonic state belong to
preparation.  Detector material interactions and reconstruction of the
conversion electron belong to readout.  Muon decay in orbit (DIO) and other processes capable of
producing a conversion-like record are separate background branches.
Table~\ref{tab:mu2e-stages} summarizes the factorization.

\begin{table}[ht]
\centering
\caption{Factorization of the Mu2e/CLFV forward model.}
\label{tab:mu2e-stages}
\small
\begin{tabularx}{\textwidth}{@{}p{0.18\textwidth}Y@{}}
\toprule
Component & Representative coordinates and functions \\
\midrule
Preparation \(u_P\) &
Production-target interactions, pion production and decay, muon transport and
selection, stopping in aluminium, atomic binding and cascade, stopped-muon
exposure and space--time distribution. \\
Phenomenon \(\mathcal X^{\mu e}_{\theta,u_X}\) &
Neutrinoless muon--electron conversion, its selected CLFV parameters, elastic
and inelastic nuclear responses, and conversion final-state energies. \\
Readout \(u_M\) &
Post-conversion electron propagation through target and detector material,
tracker momentum scale, resolution and tails, tracking efficiency, timing
response, and cosmic-ray-veto response. \\
Analysis \(u_A\) &
Momentum--time selections, binning, live gate, event classes, and the summary
constructed for inference. \\
Backgrounds \(\lambda_b,u_b\) &
DIO and other muon-, beam-, or cosmic-induced processes that can produce an
accepted conversion-like record; each has its own physical model and
normalization. \\
\bottomrule
\end{tabularx}
\end{table}

This separation affects the physics information retained in the observable.
The momenta quoted above are conversion momenta before detector
energy loss and reconstruction, whereas the Run-I selection is imposed on
reconstructed momentum.  After detector response, substantial portions of the
first two inelastic components populate the region around and below the lower
boundary of the standard search window, while the third lies farther below it.
Extending the reconstructed-momentum range therefore retains more of the
operator-sensitive inelastic structure, together with additional DIO
background and response-model dependence.

\subsection{Forward model and nuclear-response parameterization}
\label{subsec:mu2e-forward}

Let \(\Delta_i\) be a reconstructed momentum--time bin and define
\(f_{{\rm cap},A}=\Gamma_{{\rm cap},A}/\Gamma_{{\rm tot},A}\) and
\(f_{{\rm dec},A}=\Gamma_{{\rm dec},A}/\Gamma_{{\rm tot},A}\).  A binned
expected-count model is
\begin{align}
 \mu_i(\theta,u)
 &=\sum_f s_{if}(\theta,u_P,u_X,u_M,u_A)
   +b_i^{\rm DIO}(u_P,u_M,u_A,u_{\rm DIO})
   +\sum_{b\ne{\rm DIO}}b_{ib}(u_P,u_M,u_A,u_b),
 \label{eq:mu2e-counts}\\
 s_{if}
 &=N_\mu^{\rm stop}(u_P)f_{{\rm cap},A}
   R_f(\theta,u_X)\,\varepsilon_{if}(u_P,u_X,u_M,u_A),
 \label{eq:mu2e-signal-counts}
\end{align}
where \(R_f=\Gamma_f/\Gamma_{\rm cap}\).  Writing
\(y=(p_{\rm rec},t_{\rm rec})\), the efficiency may be represented schematically
as
\begin{equation}
 \varepsilon_{if}
 =\int_{\Delta_i}dy\,K_f(y;u_P,u_X,u_M)\,a_f(y;u_A),
 \label{eq:mu2e-efficiency}
\end{equation}
where \(K_f\) is obtained by applying the readout model to conversion final
state \(f\) and averaging over the stopped-muon ensemble supplied by
\(\mathcal P\).  Its dependence on \(u_P\) enters through that ensemble,
while \(u_X\) fixes the relevant final-state kinematics; material transport and
detector response enter through \(u_M\).  The factor \(a_f\) contains the
remaining classical acceptance.  A detailed detector simulation can be
substituted for \(K_f\); no additional detector formalism is required by the
TOM construction.

The DIO background is a distinct physical process.  If
\(d_{\rm DIO}(q;u_{\rm DIO})\) denotes its normalized true-momentum spectrum,
then
\begin{equation}
 b_i^{\rm DIO}
 =N_\mu^{\rm stop}f_{{\rm dec},A}
  \int dq\,d_{\rm DIO}(q;u_{\rm DIO})
  \int_{\Delta_i}dy\,K_{\rm DIO}(y\mid q;u_P,u_M)a_{\rm DIO}(y;u_A).
 \label{eq:mu2e-dio-counts}
\end{equation}
Near its endpoint the true DIO spectrum falls steeply, approximately as the
fifth power of the distance from the endpoint before convolution with the
response \cite{Czarnecki2011}.  Consequently, small momentum-scale and tail
changes can cause large changes in signal-window leakage.  In the Run-I study,
a \(100\,\mathrm{keV}/c\) momentum-scale shift changed the estimated DIO
background by \(+59\%/-37\%\) \cite{Mu2eRunI2023}.  This provides a momentum-response deformation with which to test the local
parameter-equivalence construction.

\subsection{Worked local model}
\label{subsec:mu2e-worked}

For an analytic demonstration, fix a benchmark operator and nuclear model, and
combine its selected excited-state contributions into one inelastic
template \(\boldsymbol h_1\).  Let \(\boldsymbol h_0\) and \(\boldsymbol h_1\)
be expected-bin yields per unit elastic and selected-inelastic conversion rate.
The phenomenon-model coordinate \(u_X\) is held fixed in this demonstration;
experimental variations enter the templates through
\(u_{\rm exp}=(u_P,u_M,u_A)\).  The fitted physics coordinates are
\(\theta=(R_0,\rho)\), with \(\rho=\log(R_{\rm inel}/R_0)\).  The
background-inclusive prediction is
\begin{equation}
 \boldsymbol\mu(R_0,\rho,u)
 =R_0\left[\boldsymbol h_0(u_{\rm exp})
          +e^\rho\boldsymbol h_1(u_{\rm exp})\right]
  +\boldsymbol b(u_P,u_M,u_A,u_{\rm bg}).
 \label{eq:two-template}
\end{equation}
This restricted two-parameter model is used to study local degeneracy.  The
statements below assume that \(R_0\) and \(\rho\) are freely fitted and that their
template derivatives are linearly independent.

Equation~\eqref{eq:two-template} gives
\begin{equation}
 J_{R_0}=\boldsymbol h_0+e^\rho\boldsymbol h_1,
 \qquad
 J_\rho=R_0e^\rho\boldsymbol h_1.
 \label{eq:mu2e-parameter-jacobians}
\end{equation}
It yields three parameter-degeneracy results.

\medskip
\noindent\textbf{Absolute-rate degeneracy.}\ 
After background subtraction or profiling, a one-bin elastic signal model has
\begin{equation}
 \mu_{\rm sig}=N_\mu^{\rm stop}f_{\rm cap}R_{\mu e}\varepsilon_0.
 \label{eq:one-bin-signal}
\end{equation}
For a fixed observed signal yield, refitting with changed normalization inputs
gives
\begin{equation}
 \frac{\delta\widehat R_{\mu e}}{\widehat R_{\mu e}}
 =-\frac{\delta N_{\mu,{\rm fit}}^{\rm stop}}{N_\mu^{\rm stop}}
  -\frac{\delta f_{{\rm cap},{\rm fit}}}{f_{\rm cap}}
  -\frac{\delta\varepsilon_{0,{\rm fit}}}{\varepsilon_0}.
 \label{eq:rate-fixed-data}
\end{equation}
A single count cannot distinguish these directions from a change of
\(R_{\mu e}\); control samples or external constraints are required.

\medskip
\noindent\textbf{Common signal mode.}\ 
Let \(u_+\) be a fractional efficiency change common to the elastic and
inelastic signal templates.  Conditional on separately constrained or profiled
backgrounds,
\begin{equation}
 J_{u_+}=R_0J_{R_0}.
 \label{eq:common-signal-direction}
\end{equation}
The change is visible in the predicted spectrum and exactly degenerate with the
fitted rate \(R_0\), but it gives \(\delta\widehat\rho=0\).  A stopped-muon
exposure change has the same property only after its correlated effects on DIO
and other stop-normalized backgrounds are treated consistently.

\medskip
\noindent\textbf{Relative signal mode.}\ 
Let \(u_-\) be a relative efficiency deformation,
\(\boldsymbol h_1\mapsto e^{u_-}\boldsymbol h_1\), that changes the inelastic
component but not the elastic one.  At \(u_-=0\),
\begin{equation}
 J_{u_-}=J_\rho,
 \qquad
 \delta\widehat\rho_{\rm fixed\ summary}=-\delta u_-.
 \label{eq:exact-relative-mimic}
\end{equation}
This is an exact first-order parameter-equivalent deformation.  More exposure cannot
remove it; the relative response must be constrained by calibration, simulation,
or auxiliary data.  Under the true--fit convention,
\(\Delta\rho_{\rm sys}=u_{-,{\rm true}}-u_{-,{\rm fit}}\).

\subsection{Partial parameter degeneracy from momentum response}
\label{subsec:mu2e-projection}

A momentum-scale or resolution change moves events among elastic lines,
inelastic lines, and the steep DIO spectrum.  Its effect can be separated into a
component that mimics \((R_0,\rho)\) and a residual shape component.  To display
the part relevant to \(\rho\) after profiling the freely fitted normalization,
take background parameters as fixed or externally constrained and define
\begin{equation}
 P_\perp
 =I-J_{R_0}(J_{R_0}^{\top}WJ_{R_0})^{-1}J_{R_0}^{\top}W.
 \label{eq:profile-normalization-projector}
\end{equation}
For a fixed summary and one model coordinate \(u\),
\begin{equation}
 \delta\widehat\rho
 =-\frac{(P_\perp J_\rho)^{\top}W(P_\perp J_u)}
         {(P_\perp J_\rho)^{\top}W(P_\perp J_\rho)}\,\delta u.
 \label{eq:mu2e-rho-shift}
\end{equation}
A common signal normalization satisfies \(P_\perp J_{u_+}=0\), whereas a
relative response deformation does not.  For momentum scale and resolution,
Eqs.~\eqref{eq:parameter-projector}--\eqref{eq:degeneracy-ratio} and
\eqref{eq:profile-normalization-projector}--\eqref{eq:mu2e-rho-shift} provide
the fitted parameter shift, the parameter-like fraction \(\eta\), and the
residual \((I-\Pi_\theta)J_u\delta u\).  The reconstructed templates determine
the degree of degeneracy for each response variation.

\begin{table}[ht]
\centering
\caption{Parameter-degeneracy results for two Mu2e analysis summaries.  The
momentum-response entry uses the Run-I-based benchmark of
Sec.~\ref{subsec:mu2e-numerical}; there \(R_0\) is rescaled as
\(r=R_0/10^{-15}\) for numerical conditioning, without changing the tangent
space or \(\eta\).}
\label{tab:mu2e-diagnostics}
\small
\begin{tabularx}{\textwidth}{@{}p{0.29\textwidth}YY@{}}
\toprule
Model deformation & One signal-window count & Extended elastic--inelastic spectrum \\
\midrule
Common signal efficiency &
Exactly rate-like; biases \(R_{\mu e}\) if fixed &
Exactly absorbed by \(R_0\); \(\rho\) is robust \\
Relative inelastic efficiency &
Not identifiable as a separate direction &
Exactly degenerate with \(\rho\) at first order \\
Momentum scale &
Compressed into signal acceptance and DIO leakage &
Partially parameter-equivalent; for the Run-I benchmark,
\(\eta_p=0.331\) and a substantial residual shape remains \\
DIO endpoint-spectrum model &
Background-rate direction &
Can gain internal constraint from DIO-rich lower-momentum bins, subject to
response-model identifiability and external calibration \\
Analysis threshold or binning &
Changes the selected rate &
Changes both identifiability and the corresponding response block \(B_A\) \\
\bottomrule
\end{tabularx}
\end{table}

\subsection{Numerical example of partial parameter equivalence}
\label{subsec:mu2e-numerical}

We now evaluate the projection for a Mu2e Run-I momentum-scale deformation and
test whether the resulting spectral change can be represented locally by a
change of the two physics parameters \((R_0,\rho)\).

The baseline reconstructed-momentum distributions are taken from the published
Mu2e Run-I study \cite{Mu2eRunI2023}.  Figure~20 of that work gives the
conversion-electron (CE), DIO, cosmic-ray, antiproton, and radiative-pion-capture
(RPC) spectra in \(50\,\mathrm{keV}/c\) bins for \(6\times10^{16}\) stopped
muons and the time selection \(640<T_0<1650\,\mathrm{ns}\).  The CE spectrum is
normalized to \(R_{\mu e}=10^{-15}\).  The bin contents were extracted directly
from the vector curves in the published figure.  Integration over the published
search window \(103.60<p_{\rm rec}<104.90\,\mathrm{MeV}/c\) gives \(4.085\) CE
events, compared with \(4.17\) inferred from the published single-event
sensitivity.  The same extraction gives \(0.0375\) DIO, \(0.0467\) cosmic-ray,
\(0.0110\) antiproton, and \(0.0116\) RPC events, in agreement with the
corresponding published values \(0.038\), \(0.046\), \(0.010\), and
approximately \(0.011\) \cite{Mu2eRunI2023}.

For the inelastic component we choose the isoscalar transverse-spin benchmark
of Ref.~\cite{HaxtonRule2025}.  For the three transitions introduced in
Sec.~\ref{subsec:mu2e-question}, Table~IV of that reference gives the central
relative strengths \(r_f=(0.20,0.22,0.30)\), so that
\begin{equation}
 \frac{R_{\rm inel}}{R_0}=0.72,
 \qquad
 \rho_0=\log(0.72)=-0.3285.
 \label{eq:numerical-rho0}
\end{equation}
Because \(\boldsymbol h_1\) in Eq.~\eqref{eq:two-template} denotes the template
per unit total selected-inelastic rate, we set
\(r_{\rm inel}=\sum_f r_f=0.72\), define the normalized internal weights
\(w_f=r_f/r_{\rm inel}\), and construct
\begin{equation}
 \boldsymbol h_1=\sum_f w_f\,T_f[\boldsymbol h_0],
 \label{eq:numerical-h1}
\end{equation}
where \(T_f\) translates the elastic reconstructed template to the
conversion-electron momentum of transition \(f\), following the construction
of Ref.~\cite{HaxtonRule2025}.  Consequently, the nominal inelastic contribution
is
\begin{equation}
 e^{\rho_0}\boldsymbol h_1
 =\sum_f r_f\,T_f[\boldsymbol h_0],
 \label{eq:numerical-h1-nominal}
\end{equation}
so that the total inelastic-to-elastic ratio is applied only once.

For numerical conditioning we use
\[
 r=\frac{R_0}{10^{-15}},
 \qquad
 \theta=(r,\rho),
\]
so that the nominal point is \(r_0=1\).  The binned prediction over
\(102.7\lesssim p_{\rm rec}\lesssim105.0\,\mathrm{MeV}/c\) contains the
elastic and inelastic signals together with the published DIO, cosmic-ray,
antiproton, and RPC backgrounds.  The local metric is the Asimov Poisson weight evaluated at the nominal
point,
\begin{equation}
 W_{ii}=\frac{1}{\mu_i(\theta_0,u_0)}.
 \label{eq:numerical-poisson-weight}
\end{equation}
The matrix \(W\) is held fixed when the finite-difference derivatives and
projections are evaluated, as assumed in Sec.~\ref{subsec:local-estimator}.

Let \(u_p\) denote an additive reconstructed-momentum-scale coordinate,
\(p_{\rm rec}\mapsto p_{\rm rec}+u_p\).  To estimate the local derivative we
use a symmetric \(50\,\mathrm{keV}/c\) half-step,
\begin{equation}
 J_p
 \simeq
 \frac{\boldsymbol\mu(u_p=+0.05\,\mathrm{MeV}/c)
       -\boldsymbol\mu(u_p=-0.05\,\mathrm{MeV}/c)}
      {0.10\,\mathrm{MeV}/c}.
 \label{eq:numerical-scale-derivative}
\end{equation}
The \(100\,\mathrm{keV}/c\) value quoted below is instead the physical benchmark
mismatch to which this local derivative is applied.  Repeating the derivative with \(100\) and \(25\,\mathrm{keV}/c\) half-steps
changes the response coefficients and \(\eta_p\) by at most about \(2\%\).
Across the nominal range and ranges obtained by trimming up to three outer bins,
\(0.315\leq\eta_p\leq0.331\), so the conclusion is stable under these checks.
The resulting TOM response vector is
\begin{equation}
 B_p=J_{\theta,W}^{+}J_p
 =
 \begin{pmatrix}
  0.868\\[-1mm]
 -0.800
 \end{pmatrix}
 (\mathrm{MeV}/c)^{-1}.
 \label{eq:numerical-Bp}
\end{equation}
Thus, with the true--fit sign convention of
Eq.~\eqref{eq:true-fit-parameter-shift}, a
\(+100\,\mathrm{keV}/c\) mismatch of the true momentum scale relative to the
model used in the fit gives
\begin{equation}
 \Delta r=+0.0868,
 \qquad
 \Delta\rho=-0.0800.
 \label{eq:numerical-parameter-shift}
\end{equation}
The inferred elastic normalization therefore increases by about \(8.7\%\),
whereas the inferred inelastic-to-elastic ratio changes by
\(\exp(-0.0800)-1\simeq-7.7\%\).  To first order the absolute inelastic rate
changes only by
\(\Delta R_{\rm inel}/R_{\rm inel}\simeq\Delta r+\Delta\rho\simeq0.7\%\).
In this benchmark the scale mismatch is interpreted mainly as an increase of
the elastic contribution and a corresponding reduction of the
inelastic-to-elastic ratio; the absolute inelastic strength changes little.

The parameter shift reproduces only part of the spectral deformation.  The
projection ratio is
\begin{equation}
 \eta_p
 =
 \frac{\|\Pi_\theta J_p\|_W}{\|J_p\|_W}
 =0.331,
 \qquad
 \frac{\|(I-\Pi_\theta)J_p\|_W}{\|J_p\|_W}
 =0.944.
 \label{eq:numerical-eta}
\end{equation}
Because the two components are \(W\)-orthogonal, \(\eta_p^2\simeq0.110\):
only about \(11\%\) of the squared weighted deformation is absorbed by the
two-dimensional physics-parameter tangent space, while about \(89\%\) remains
outside it.  Figure~\ref{fig:scale-decomposition} shows the corresponding
first-order decomposition for a \(+100\,\mathrm{keV}/c\) displacement.

\begin{figure}[ht]
 \centering
 \includegraphics[width=0.88\textwidth]{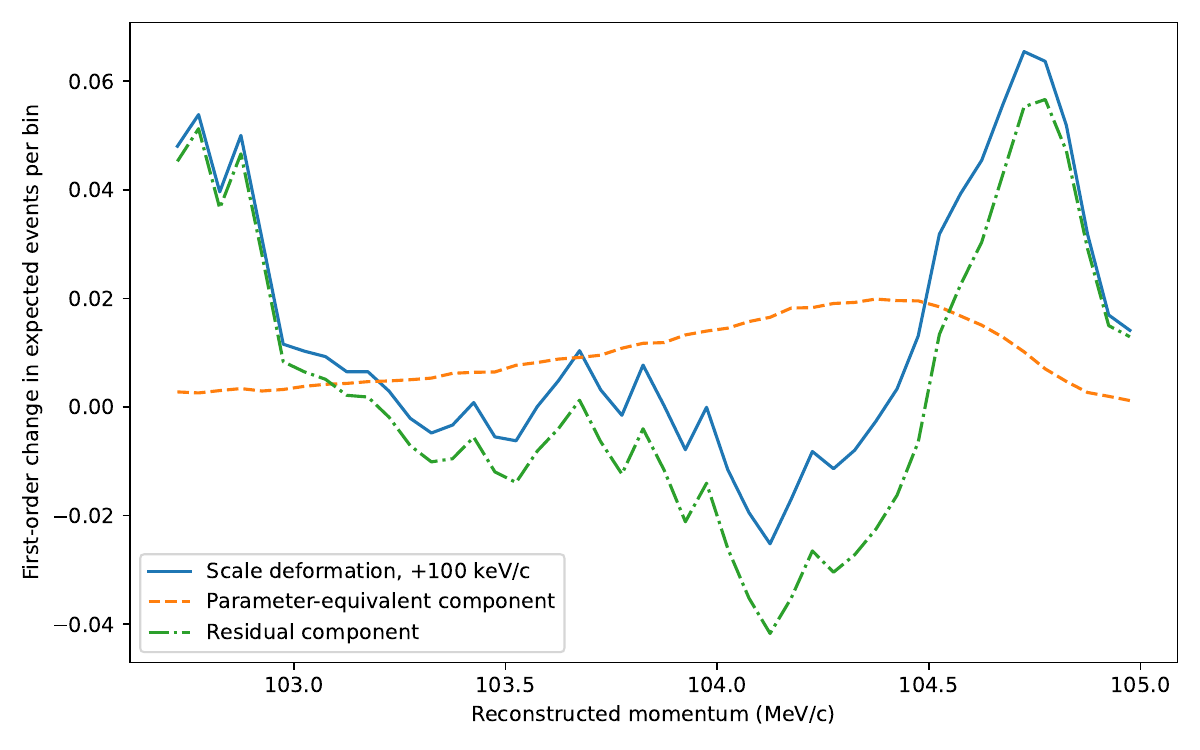}
 \caption{First-order decomposition of a
 \(+100\,\mathrm{keV}/c\) reconstructed-momentum-scale deformation in the
 Run-I-based numerical model.  The parameter-equivalent component is the
 \(W\)-orthogonal projection onto the tangent space generated by variations of
 \(r\) and \(\rho\); the residual component cannot be reproduced by any
 infinitesimal change of these two physics parameters.  Baseline signal and
 background bin contents were extracted from the vector curves in Fig.~20 of
 Ref.~\cite{Mu2eRunI2023}}
 \label{fig:scale-decomposition}
\end{figure}

Here ``not parameter-equivalent'' has a geometric, rather than a statistical,
meaning.  It means that no infinitesimal change
\((\delta r,\delta\rho)\) can reproduce the complete scale-induced spectral
change: \(J_p\delta u_p\notin\operatorname{col}(J_\theta)\).  The projected
part nevertheless shifts the fitted physics parameters according to
Eq.~\eqref{eq:numerical-parameter-shift}; the orthogonal part remains as a
shape mismatch that cannot be absorbed by refitting \(r\) and \(\rho\).
An exactly parameter-equivalent deformation would have zero residual, in which
case increasing the statistics of the same observable would not break the
local degeneracy.

A nonzero residual does not, however, imply that the experiment has enough
statistical power to detect it.  For the Run-I exposure and the fiducial
\(R_0=10^{-15}\) benchmark used here, the quadratic Asimov size of the
remaining deformation is only
\begin{equation}
 \Delta\chi^2_{\rm res}
 =
 \left\|0.10\,(\mathrm{MeV}/c)\,
 (I-\Pi_\theta)J_p\right\|_W^2
 \simeq0.25.
 \label{eq:numerical-residual-chi2}
\end{equation}
Thus the residual is physically distinguishable in principle, because it lies
outside the physics-parameter tangent space, but it would not by itself
constitute a statistically significant lack of fit at this Run-I benchmark.
This separates \emph{identifiability of the deformation} from
\emph{statistical detectability of its residual}.  The first is fixed by the
local model geometry; the second also depends on exposure, signal strength,
backgrounds, and the sampling distribution.

As a separate validation of the numerical deformation, shifting the extracted
DIO search window by \(\mp100\,\mathrm{keV}/c\) changes the integrated DIO
yield by \(+58.9\%\) and \(-37.0\%\), respectively, reproducing the published
Mu2e estimate \(+59\%/-37\%\) \cite{Mu2eRunI2023}.  The agreement provides an
independent check that the extracted spectrum and momentum-scale convention
used in Eq.~\eqref{eq:numerical-scale-derivative} reproduce the relevant Run-I
response.  All numerical inputs used here are traceable to published sources: the
baseline Run-I reconstructed-momentum spectra are extracted from Fig.~20 of
Ref.~\cite{Mu2eRunI2023}, while the excited-state kinematics and relative
inelastic strengths are taken from Ref.~\cite{HaxtonRule2025}.  The
translation, interpolation, boundary treatment, finite-difference convention,
and projection procedure are specified above.

\section{Discussion and conclusions}
\label{sec:discussion}

The reformulated TOM casts robustness as a relation between deformations of a complex forward
model and the local tangent space generated by the physics parameters.  The
stage decomposition retains the physical origin of each deformation, while
\(\Pi_\theta\) separates its parameter-equivalent component from the residual
left outside that tangent space.

Exact parameter equivalence occurs when the deformation lies entirely in
\(\col(J_\theta)\).  Partial equivalence leaves a residual spectral direction
that cannot be absorbed by refitting the selected physics parameters.  The same
criterion applies to preparation, phenomenon modelling, readout, analysis, and
background coordinates.

The \(\mu^-\!\to e^-\) example places this construction in CLFV phenomenology.  A one-bin conversion rate is exactly degenerate with
its normalization inputs.  In a restricted positive-signal model that retains
elastic and inelastic nuclear responses, a common signal normalization changes
the fitted rate but not the operator-sensitive response ratio, whereas a
relative elastic--inelastic efficiency deformation is exactly equivalent, at
first order, to a change of that ratio.  The reconstructed-momentum-scale example supplies a nontrivial numerical
case: a \(100\,\mathrm{keV}/c\) mismatch has
\(\eta_p=0.331\), so it biases the inferred parameters but is far from exact
parameter equivalence.  Most of its weighted spectral deformation lies outside
the \((R_0,\rho)\) tangent space, although at the Run-I benchmark the remaining
shape mismatch is not itself statistically significant.  The example does not imply that a single
aluminium target can determine a general CLFV effective theory; broader
operator discrimination requires additional targets or complementary CLFV
processes \cite{Cirigliano2009}.

The result is local.  It assumes a smooth forward model, an identifiable
parameterization, a fixed weight matrix, and a neighbourhood of a nominal fit.  Large or discrete analysis changes,
nonregular likelihoods near physical boundaries, finite-sample bias, and
weakly constrained jointly fitted nuisances require direct likelihood studies.
These conditions define the regime in which \(B\) and \(\Pi_\theta\) act as
local maps between deformations of the complete model and directions in the
physics-parameter manifold.

Within those limits, TOM gives a criterion for determining when a change of
experimental realization or auxiliary modelling is locally indistinguishable,
at the level of specified observables, from a change of the inferred theory
parameters.  For partial equivalence it also identifies the residual that
cannot be absorbed by those parameters.  This provides a reusable way to ask
which additional observables or constraints are needed to break a local
physics--model degeneracy.

\section*{Data sources and reproducibility}
No new experimental data were generated or analysed.  The numerical
illustration uses projected Mu2e Run-I simulation spectra published in
Ref.~\cite{Mu2eRunI2023}.  In particular, the baseline conversion-electron and
background spectra are extracted from the vector curves in Fig.~20 of that
reference.  The excited-state conversion momenta and the isoscalar
transverse-spin relative strengths used to construct the inelastic template
are taken from Ref.~\cite{HaxtonRule2025} (Tables~II and IV).  The numerical
construction, finite-difference convention, validation checks, and local
projection procedure are described in Sec.~\ref{subsec:mu2e-numerical}.  The
derived bin tables and reproduction script are available from the author.

\section*{Acknowledgments}
This work was carried out under contract No.~89243024CSC000002 with the
U.S. Department of Energy, Office of Science, Office of High Energy Physics.


\begin{thebibliography}{99}
\bibitem{Pronskikh2020}
V.~S. Pronskikh,
``Measurement problems: Contemporary discussions and models,''
\emph{Phys. Usp.} \textbf{63}, 192--200 (2020),
\url{https://doi.org/10.3367/UFNe.2019.06.038583}.

\bibitem{Fock1957}
V.~A. Fock,
``On the interpretation of quantum mechanics,''
\emph{Czech. J. Phys.} \textbf{7}, 643--656 (1957),
\url{https://doi.org/10.1007/BF01946586}.

\bibitem{Busch1995}
P.~Busch, M.~Grabowski, and P.~J. Lahti,
\emph{Operational Quantum Physics},
Springer, Berlin (1995),
\url{https://doi.org/10.1007/978-3-540-49239-9}.

\bibitem{Cook1986}
R.~D. Cook,
``Assessment of local influence,''
\emph{J. R. Stat. Soc. B} \textbf{48}, 133--155 (1986),
\url{https://doi.org/10.1111/j.2517-6161.1986.tb01398.x}.

\bibitem{Watrous2018}
J.~Watrous,
\emph{The Theory of Quantum Information},
Cambridge University Press, Cambridge (2018),
\url{https://doi.org/10.1017/9781316848142}.

\bibitem{Kay1993}
S.~M. Kay,
\emph{Fundamentals of Statistical Signal Processing, Volume I: Estimation
Theory}, Prentice Hall, Upper Saddle River (1993).

\bibitem{CowanEtAl2011}
G.~Cowan, K.~Cranmer, E.~Gross, and O.~Vitells,
\emph{Asymptotic formulae for likelihood-based tests of new physics},
\emph{Eur. Phys. J. C} \textbf{71}, 1554 (2011); Erratum
\textbf{73}, 2501 (2013),
\url{https://doi.org/10.1140/epjc/s10052-011-1554-0}.

\bibitem{KunoOkada2001}
Y.~Kuno and Y.~Okada,
``Muon decay and physics beyond the standard model,''
\emph{Rev. Mod. Phys.} \textbf{73}, 151--202 (2001),
\url{https://doi.org/10.1103/RevModPhys.73.151}.

\bibitem{Kitano2002}
R.~Kitano, M.~Koike, and Y.~Okada,
``Detailed calculation of lepton flavor violating muon--electron conversion
rate for various nuclei,''
\emph{Phys. Rev. D} \textbf{66}, 096002 (2002); Erratum
\emph{Phys. Rev. D} \textbf{76}, 059902 (2007),
\url{https://doi.org/10.1103/PhysRevD.66.096002},
\url{https://doi.org/10.1103/PhysRevD.76.059902}.

\bibitem{Mu2eTDR2015}
L.~Bartoszek \emph{et al.} (Mu2e Collaboration),
``Mu2e Technical Design Report,'' Fermilab-TM-2594,
Fermilab-DESIGN-2014-1 (2015),
\url{https://doi.org/10.2172/1172555}.

\bibitem{Cirigliano2009}
V.~Cirigliano, R.~Kitano, Y.~Okada, and P.~Tuzon,
``On the model discriminating power of $\mu\to e$ conversion in nuclei,''
\emph{Phys. Rev. D} \textbf{80}, 013002 (2009),
\url{https://doi.org/10.1103/PhysRevD.80.013002}.

\bibitem{HaxtonRule2024}
W.~C. Haxton and E.~Rule,
``Distinguishing charged lepton flavor violation scenarios with inelastic
$\mu\to e$ conversion,''
\emph{Phys. Rev. Lett.} \textbf{133}, 261801 (2024),
\url{https://doi.org/10.1103/PhysRevLett.133.261801}.

\bibitem{HaxtonRule2025}
W.~C. Haxton and E.~Rule,
``Nuclear-level effective theory of $\mu\to e$ conversion: Inelastic
process,''
\emph{Phys. Rev. C} \textbf{111}, 025501 (2025),
\url{https://doi.org/10.1103/PhysRevC.111.025501}.

\bibitem{RuleEtAl2023}
E.~Rule, W.~C. Haxton, and K.~McElvain,
``Nuclear-level effective theory of $\mu\to e$ conversion,''
\emph{Phys. Rev. Lett.} \textbf{130}, 131901 (2023),
\url{https://doi.org/10.1103/PhysRevLett.130.131901}.

\bibitem{HaxtonEtAl2023}
W.~C. Haxton, E.~Rule, K.~McElvain, and M.~J. Ramsey-Musolf,
``Nuclear-level effective theory of $\mu\to e$ conversion: Formalism and
applications,''
\emph{Phys. Rev. C} \textbf{107}, 035504 (2023),
\url{https://doi.org/10.1103/PhysRevC.107.035504}.

\bibitem{Mu2eRunI2023}
F.~Abdi \emph{et al.} (Mu2e Collaboration),
``Mu2e Run I sensitivity projections for the neutrinoless
$\mu^-\to e^-$ conversion search in aluminum,''
\emph{Universe} \textbf{9}, 54 (2023),
\url{https://doi.org/10.3390/universe9010054}.

\bibitem{Czarnecki2011}
A.~Czarnecki, X.~Garcia i Tormo, and W.~J. Marciano,
``Muon decay in orbit: Spectrum of high-energy electrons,''
\emph{Phys. Rev. D} \textbf{84}, 013006 (2011),
\url{https://doi.org/10.1103/PhysRevD.84.013006}.

\end{thebibliography}
\end{document}